\documentclass[letterpaper, 10 pt, conference]{ieeeconf}  

\usepackage{booktabs}  
\usepackage{svg}
\usepackage{mathtools}
\usepackage{graphicx}
\usepackage{booktabs}
\usepackage{amssymb,amsfonts,mathrsfs,amsmath}
\usepackage{multicol}
\usepackage{multirow}
\usepackage{algorithmic}
\usepackage{algorithm}
\usepackage{array}
\usepackage[caption=false,font=normalsize,labelfont=sf,textfont=sf]{subfig}
\usepackage{textcomp}
\usepackage{stfloats}
\usepackage{url}
\usepackage{verbatim}
\usepackage{balance}
\usepackage{cite}
\usepackage{colortbl}
\usepackage{bm}
\usepackage{hyperref}

\usepackage{siunitx}

\newcommand{\R}{\mathbb{R}}
\newcommand{\pp}{\mathbb{P}}

\newcommand{\ff}{\mathcal{F}}
\newcommand{\ppp}{\mathbf{P}}
\newcommand{\eee}{\mathbf{E}}
\newcommand{\X}{\mathcal{X}}
\newcommand{\K}{\mathcal{K}}
\newcommand{\cj}{\wedge}
\renewcommand{\L}{\mathcal{L}}
\renewcommand{\preceq}{\lesssim}
\newcommand{\rr}{\mathcal{R}}
\newcommand{\trans}{{\mathsf{T}}}
\newcommand{\eps}{\varepsilon}
\newtheorem{thm}{\bf Theorem}[section]

\newtheorem{rem}[thm]{\bf Remark}

\newtheorem{prop}[thm]{\bf Proposition}

\newtheorem{deff}[thm]{\bf Definition}

\newcommand{\set}[1]{\left\{#1\right\}}
\newcommand{\Qed}{\hfill$\Box$} 

\usepackage{mathptmx}
\usepackage{dsfont}
\DeclareMathAlphabet{\mathcal}{OMS}{cmsy}{m}{n}
\IEEEoverridecommandlockouts                              

\title{\LARGE \bf
Koopman Spectral Analysis and System Identification for Stochastic Dynamical Systems via Yosida Approximation of Generators
}

\author{Jun Zhou$^{\star}$, Yiming Meng$^{\star}$, and Jun Liu
\thanks{$^{\star}$ Equal contribution}
\thanks{Jun Zhou is with the School of Artificial Intelligence and Automation \& the China-EU Institute for Clean And Renewable Energy, Huazhong University of Science and Technology, Wuhan, China.
{\tt\small junzhou02@hust.edu.cn}}
\thanks{
Yiming Meng is with  the
Coordinated Science Laboratory, University of Illinois Urbana-Champaign,
Urbana, IL 61801, USA.
        {\tt\small ymmeng@illinois.edu}}
\thanks{
Jun Liu is with the Department of Applied Mathematics,
University of Waterloo, Waterloo, ON N2L 3G1, Canada. {\tt\small j.liu@uwaterloo.ca}.}
}        

\begin{document}
\newcommand{\ZJ}[1]
{\textcolor{red}{\bf (ZJ:  #1)}}

\maketitle
\thispagestyle{empty}
\pagestyle{empty}

\begin{abstract}
System identification and Koopman spectral analysis are crucial for uncovering physical laws and understanding the long-term behaviour of stochastic dynamical systems governed by stochastic differential equations (SDEs). In this work, we propose a novel method for estimating the Koopman generator of systems of SDEs, based on the theory of resolvent operators and the Yosida approximation. This enables both spectral analysis and accurate estimation and reconstruction of system parameters. The proposed approach relies on only mild assumptions about the system and effectively avoids the error amplification typically associated with direct numerical differentiation. It remains robust even under low sampling rates or with only a single observed trajectory, reliably extracting dominant spectral modes and dynamic features. We validate our method on two simple systems and compare it with existing techniques as benchmarks. The experimental results demonstrate the effectiveness and improved performance of our approach in system parameter estimation, spectral mode extraction, and overall robustness.
\end{abstract}

\section{INTRODUCTION}

Stochastic dynamical systems arise throughout physics, biology, and finance as nonlinear systems subject to random perturbations\cite{chen2023stochastic}. While stochastic differential equations (SDEs) rigorously model stochastic phenomena \cite{oksendal2013stochastic}, conventional methods for system identification (e.g., sparse regression\cite{brunton2016discovering}, maximum likelihood estimation (MLE)\cite{schon2006maximum}) face well-documented challenges when handling nonlinear systems with partial and noisy observations. This issue is especially pronounced in experimental settings, where limited sampling rates and scarce trajectory data further complicate the task. Robust and efficient system identification methods are therefore crucial for dynamic analysis and control under such constraints.

Koopman operator theory \cite{koopman1931hamiltonian}, \cite{koopman1932dynamical} provides a powerful framework for analyzing nonlinear systems by mapping nonlinear dynamics into linear operations on function spaces, revealing system structure through spectral analysis and enabling data-driven prediction without explicit equations \cite{budivsic2012applied, mezic2013analysis}.
Its infinitesimal generator explicitly encodes SDE drift and diffusion terms, enabling spectral analysis of stability and modal dynamics \cite{mezic2005spectral}. While analytical determination of the generator is challenging, data-driven methods such as dynamic mode decomposition (DMD) \cite{schmid2010dynamic}, extended DMD (EDMD) \cite{williams2014kernel}, and generator EDMD (gEDMD) \cite{klus2020data} offer effective approximations. However, their reliance on high-frequency data and lack of robustness to noise remain critical limitations.

Unlike traditional approaches that estimate generators directly through temporal state derivatives or local fitting\cite{klus2020data}, \cite{brunton2016discovering}, recent works \cite{meng2024resolvent,zeng2024data, meng2024koopmanitsc, susuki2021koopman} employ a modified, data-driven version of the Yosida approximation and achieve high precision through rigorous analysis. Numerical examples demonstrate its effectiveness in deterministic non-polynomial, nonlinear, chaotic systems, and outperform the aforementioned methods. To better address the theoretical and practical challenges of generator learning, system identification, and spectral estimation for stochastic systems, this paper proposes a Koopman generator learning and spectral estimation method based on the resolvent theory and Yosida approximation framework. It is worth noting that, given the weak topology of solutions to SDEs, the potentially strong diffusion effects, and the typically bounded region of interest, this is not a straightforward extension of \cite{meng2024resolvent}. Through rigorous analysis, we demonstrate that the proposed method achieves high precision and convergence guarantees, while more intuitive or less careful treatments may overlook some out-of-domain effects. We show that the proposed method robustly extracts dominant spectral modes and dynamical features from stochastic systems even under conditions of low sampling rates or significant noise interference. This makes it particularly suitable for practical spectral analysis of stochastic systems.

The rest of this paper is organized as follows. Section \ref{sec:2} introduces preliminaries on Koopman operator for stochastic systems and formulates the problem. Section \ref{sec:3} presents the proposed spectral estimation method under the Resolvent-Yosida theoretical framework, with an emphasis on convergence analysis and a discussion of how some intuitive treatments can lead to unexpectedly large errors. The data-driven algorithm is developed in Section \ref{sec: data}, and the effectiveness of the proposed approach is demonstrated through numerical examples in Section \ref{sec:4}.

\textbf{Notation}: 
	Throughout this paper, we adopt the following notation.  The domain of an operator is denoted as $\mathcal D(\cdot)$. The linear space spanned by basis functions is denoted as $\text{span}\{\cdot\}$. 
    The Euclidean space of dimension $d > 1$ is denoted as $\mathbb{R}^d$. The Euclidean norm is denoted as $| \cdot |$. For a set $A \subseteq \mathbb{R}^d$, $\overline{A}$ denotes its closure and $\partial A$ denotes its boundary. The transpose of a matrix/vector transpose and pseudoinverse $a$ are denoted as $a^\top$ and $a^\dagger$, respectively.

For any 
stochastic processes $\{X(t)\}_{t\geq 0}$, we use the shorthand notation $X:=\{X(t)\}_{t\geq 0}$.   For any stopped process $\{X({t\cj\tau})\}_{t\geq 0}$, where $\tau$ is a stopping time and $t\cj\tau=\min(t, \tau)$, we use the shorthand notation $X^\tau$. We denote the Borel $\sigma$-algebra of a set by $\mathscr{B}(\cdot)$. We write $a\preceq b$ if there exists a  
$C>0$  (independent of $a$ and $b$)  such that $a\leq Cb$.  

\section{Preliminary and Problem formulation} \label{sec:2}

\subsection{Stochastic Koopman Operator Framework}
 We consider a 
 stochastic dynamical system 
\begin{equation} \label{eq1}
    dX(t) = f(X(t))\,dt + b(X(t))\,dW(t),
\end{equation}
where  
the state space $\mathcal{X}\subseteq\R^d$ is a bounded open domain; $W$ represents an $l$-dimensional standard Wiener process; $f:\X\rightarrow \R^d$ is a locally Lipschitz non-linear drift vector field; the diffusion term $b:\X\rightarrow\R^{d \times l}$ is smooth.

Note that from a modeling perspective, we typically do not specify a Wiener process \textit{a priori} \cite{oksendal2013stochastic}. Moreover, for verifying dynamical behaviors through probability laws, restricting to a specific probability space is unnecessary. We therefore consider the following natural solution concept. 

\begin{deff}[Weak solutions]
The system \eqref{eq1} admits a weak solution if there exists a filtered probability space $(\Omega^\dagger,\mathscr{F}^\dagger,\{\mathscr{F}^\dagger_t\},  \pp^\dagger)$, where a Wiener process $W$ is defined and a pair of processes $(X,W)$ are adapted, such that $X$ solves the SDE \eqref{eq1}. \Qed
\end{deff} 

While the base probability space $(\Omega^\dagger,\mathscr{F}^\dagger,\{\mathscr{F}^\dagger_t\},  \pp^\dagger)$ for the Wiener process $W$ remains unspecified, we transfer information to the canonical space\footnote{Define $\Omega:=C([0,\infty);\R^n) $ with coordinate process $\mathfrak{X}_t(\omega):=\omega(t)$ for all $t\geq 0$ and all $\omega\in\Omega$. Define $\ff_t:=\sigma\{\mathfrak{X}_s,\;0\leq s\leq t\}$  for each $t\geq 0$, then  the smallest $\sigma$-algebra containing the sets in every $\ff_t$, i.e.  $\ff:=\bigvee_{t\geq 0}\ff_t$, turns out to be same as $\mathscr{B}( \Omega) $. 
For a weak solution $X$ of \eqref{eq1}, the induced measure (probability law) $\ppp$ on $\ff$ is such that $\ppp\set{A}=\pp^\dagger\circ (X)^{-1}\set{A}$ for every $A\in\mathscr{B}( \Omega)$.
The canonical probability space for $X$ is then $(\Omega,\ff,  \ppp)$. We also denote $\eee$ by the associated expectation operator w.r.t. $\ppp$.} $(\Omega,\ff, \{\ff_t\}, \ppp)$. This facilitates our analysis of both the probability law of weak solutions and their state-space probabilistic behavior, while also supporting data-driven methods that rely solely on state-space solution information. We also denote by $\ppp^x\set{\cdot}=\ppp\set{\;\cdot\;|X(0)=x}$ the probability law for the solution process of system \eqref{eq1} with initial condition $X(0)=x$ a.s., and by $\eee^x$ the expectation operator w.r.t. $\ppp^x$.  


Note that solutions to \eqref{eq1} possess the Markov property. A classical approach to studying the evolution of observables in Markov processes is through the semigroup of transition operators, namely the stochastic Koopman operators. For any observable $h\in C(\X)$ and $t\ge 0$, the Koopman operator $\K_t$ is defined as $\mathcal{K}_t h(x) := \eee^x \set{h(X(t))}$ \cite{hollingsworth2008stochastic}. The family $\{\K_t\}_{t\geq 0}$ forms a semigroup, satisfying  $\K_0=I$ and  $\mathcal{K}_{t+s} = \mathcal{K}_s \circ \mathcal{K}_t$ for all $s,t\geq 0$. Moreover,   there exist  constants $w\geq 0$ and $M\geq 1$ such that $\|\K_t\|\leq M
   e^{wt}$ for all $t\geq 0$ \cite[Theorem 1.2.2]{pazy2012semigroups}.

The infinitesimal generator $\mathcal{L}$ of this semigroup is defined as $\mathcal{L}h = \lim_{t \downarrow 0^+} \frac{\mathcal{K}_t h - h}{t}$, which captures the instantaneous change rate of observable functions as the system state evolves under \eqref{eq1}. Suppose that $\X$ is bounded and $h \in    C^2(\mathcal{X})$, the generator coincides with the It\^{o}’s lemma,  given by
\begin{small}
    \begin{equation} \label{eq5}
    \mathcal{L} h = \sum_{i=1}^d f_i\frac{\partial h}{\partial x_i} + \frac{1}{2}\sum_{i,j=1}^d [bb^T]_{ij}\frac{\partial^2 h}{\partial x_i \partial x_j}. 
\end{equation}
\end{small}


To characterize the system, it is equivalent to study the entire semigroup $\{\K_t\}_{t\geq 0}$ or its generator $\L$. Note that, for system \eqref{eq1} with $X(0)=x$ a.s. and any observable   $h\in C^2(\X)$, the quantity $M_h(t):=h(X(t))-h(x)-\int_0^t\L h(X(s))ds$ forms a martingale process.   For the inverse problem of system characterization, the   standard procedure is to verify whether $M_h$ is a martingale process under a learned generator $\L$, based on the  observable  data from   $X$. 
We say   the probability law $\ppp^x$ of $X$ (governed by the learned $\L$) solves the \textit{martingale problem} if $M_h$ is a martingale under $\ppp^x$ for all admissible $h$. Consequently, $X$ is weak solution to the system with  the learned generator $\L$.

Another primary motivation for the Koopman operator framework is to enable spectral analysis of nonlinear stochastic systems through a linear operator perspective.  Specifically, we say that a function $\phi$ is a Koopman eigenfunction associated with eigenvalue $\alpha$ if it satisfies $ \mathcal{K}_t \phi = e^{\alpha t} \phi$ for all $t \ge 0$,
and consequently, $\mathcal{L} \phi = \alpha \phi$, 
whenever $\phi$ lies in the domain of the generator.  The eigenvalues characterize the dominant stochastic modes of the SDEs, where the real and imaginary parts of $\alpha$ respectively determine the exponential decay/growth rate and oscillatory frequency of each mode.  The diffusion-driven smoothing effect naturally enhances spectral separation of the principal modes, enabling both   clearer modal decomposition and   more compact representation of the underlying SDE dynamics.

\subsection{Problem Formulation}
Our objective is to identify the generator of the autonomous SDE in \eqref{eq1} through data-driven approximation of vector-valued observables $ Z_N(x) = [z_1(x), z_2(x), \dots, z_N(x)]^\trans$, where  \( z_n \in C^2(\X)\) for each $n\in\set{1, 2, \cdots, N}$.
Consequently, this approach simultaneously approximates both the drift term $f$ and the diffusion matrix  \( B = bb^\trans \). 
We also prove that the solution to the learned system converges to the true solution of \eqref{eq1} in probability law within $\X$ via a martingale problem argument. Additionally, we   characterize the spectral properties of the associated  generator $\L$.

\section{ Resolvent-Type Approximation of Stochastic System Generators} \label{sec:3}
In this section, inspired by recent advances on resolvent-type approximations of generators for deterministic systems, we introduce a novel resolvent-based approximation method for stochastic system generators. Leveraging the Yosida approximation   and the resolvent operator, we construct a sequence of approximation procedures by modifying the resolvent operator through finite truncations. This enables projection-based approximation of the generator onto a finite-dimensional dictionary. We also demonstrate that this modification is not a straightforward extension of deterministic approximation schemes, and a less careful treatment of sample path information may lead to unbounded errors.

\subsection{Preliminaries on Yosida Approximation}
The Yosida approximation provides a powerful analytical tool for approximating potentially unbounded generators $\L$ using sequences of bounded operators. 

For $\L$, we define the resolvent set as  $\rho(\L):=\set{\lambda\in\mathbb{C}: \lambda I-\L\;\text{is invertible}}$. 
  Accordingly, the resolvent operator is defined as $\rr(\lambda; \L):=(\lambda I-\L)^{-1}$ for $\lambda\in\rho(\L)$. It is well-known that for, each  $\lambda\in\rho(\L)$,   $\rr(\lambda; \L)$ is a bounded linear operator  \cite[Chap. I,  Theorem 4.3]{pazy2012semigroups}.

\textbf{Definition 3 (Yosida Approximation)}  
We define a family of operators $\mathcal{L}_{\lambda}$ as:
\begin{equation}\label{E: Yosida}
\L_\lambda:=\lambda \L \rr(\lambda;\L) = \lambda^2\rr(\lambda;\L)-\lambda I.
\end{equation}
The family of operators $\{\mathcal{L}_{\lambda}\}_{\lambda>w}$ is called the Yosida approximation of the operator $\mathcal{L}$, where $w$ is the growth bound satisfying  $\|{\mathcal{K}^t}\|\leq M
   e^{w t}$ for all $t>0$. 

These approximations converge strongly to $\mathcal{L}$  as $\lambda \to \infty$, i.e., $\lim_{\lambda \to \infty} \mathcal{L}_{\lambda}h = \mathcal{L}h$, $ \forall h \in C_b^2(\R^n)$. Additionally, the  convergence rate is of order $\mathcal{E}_1\sim\mathcal{O}(\frac{1}{\lambda-w})$. 
 
Additionally, the resolvent operator admits a Laplace transform representation through the semigroup\cite{oksendal2013stochastic}:
\begin{equation}\label{E: resolvent}
R(\lambda; \mathcal{L})h(x) = \int_{0}^{\infty} e^{-\lambda t}K_t h(x)\,\mathrm{d}t,\;h\in C_b^2(\R^2).
\end{equation} 
The inverse relationship  derived from \eqref{E: Yosida} and \eqref{E: resolvent}  enables   approximation of $\L$ via Koopman operator representations of system transitions.

\subsection{Finite-Horizon Modification of  Yosida Approximation}
The Laplace transform representation in \eqref{E: resolvent} is not well-suited for data-driven approximation for two key reasons:  1) the infinite-horizon integration is incompatible with finite-time data collection; and 2) the combination of \eqref{E: Yosida} and \eqref{E: resolvent} only guarantees approximation of  $\L$ within $C_b^2(\R^n)$. However, the dictionary   functions $Z_N(x)$ may be unbounded over this domain, which would in turn require data collection over an unbounded spatial region.

In this section, we develop a finite-horizon approximation method for the resolvent operator $\rr(\lambda;\L)$ and its associated Yosida approximation $\L_\lambda$. 
It is important to note that, due to the intrinsic It\^{o} diffusion properties of solutions to \eqref{eq1}, finite-horizon truncations of $\rr(\lambda;\L)$  for stochastic systems exhibit fundamentally different behavior from their deterministic counterparts as in \cite{meng2024resolvent}. A crucial consideration involves out-of-domain transitions, where naive extensions of deterministic approaches \cite{meng2024resolvent} may lead to unbounded approximation errors due to diffusion. 

For simplicity, we can assume that   $\X:=\{x\in\R^n: |x|< R\}$ for some $R>0$. Let $\tau=\inf\{t>0: X_t\notin \X\}$. It is clear that $\tau\in(0, \infty)$ $\ppp^x$-a.s. for all $x\in\X$. We then work on the stopped process  $X^\tau:=\{X_{t\cj\tau}\}_{t\geq 0}$, which   satisfies $X^\tau(t)=X(t)$ for $t\in[0, \tau)$. Therefore, to recover the transitions and dynamics within $\X$, it is equivalent to study $X^\tau$. It can also be verified that, for all $h\in C^2(\R^n)$, 
\begin{small}
    \begin{equation}
    \begin{split}
        h(X^\tau(t))   = h(X({t\cj\tau}))  =h(X(t))\mathds{1}_{\{t<\tau\}}+h(X(\tau))\mathds{1}_{\{t\geq\tau\}}. 
    \end{split}
\end{equation}
\end{small}

On the other hand, denoting $\L^\tau$ by the generator for $X^\tau$, we can verify that there exists a function  $c\in C(\R^n)$ such that  $c(x)=0$ for $x\in\X$ and $c(x)>0$ otherwise, ensuring   
  $\L^\tau h(x)=\L h(x) - c(x)h(x)$ for all $h\in C^2(\R^n)$. It is also clear that $\L^\tau = \L$ when restricted to the domain $\X$. The Koopman operators $K^\tau_t$ generated by $\L^\tau$  have the following property, which can be demonstrated using the famous Feynman-Kac formula and the `killing' of diffusion processes \cite{karlin1981second}.

  \begin{prop}
      For all $t\in (0, \infty)$ and all $h\in C(\R^n)$,  
                \begin{equation}
    \begin{split}
        \K^\tau_t h(x) & =\eee^x[h(X^\tau(t))]\\
        &=\eee^x[h(X(t))\mathds{1}_{\{t<\tau\}}]+\eee^x[h(X(\tau))\mathds{1}_{\{t\geq\tau\}}]\\
        & = \eee^x\left[h(X_t)\cdot\exp\left\{-\int_0^t c(X_s)ds\right\}\right]. 
    \end{split}
\end{equation}
  \end{prop}
Then, the resolvent operator for $\L^\tau$ is defined accordingly as $\rr(\lambda;\L^\tau) h(x)  =   \int_0^\infty e^{-\lambda t} \eee^x[h(X^\tau(t)]dt$. The Yosida approximation for $\L^\tau$ can also be defined analogously to \eqref{E: Yosida} as $\L^\tau_\lambda:=\lambda^2\rr(\lambda;\L^\tau)-\lambda I$ with the same convergence result.

For the purpose of a finite-horizon approximation for $\rr(\lambda;\L^\tau)$, we introduce 
\begin{small}
    \begin{equation}\label{E: R_trunc}
   \rr^\tau_{\lambda, T}h(x) :=\int_0^{T\cj\tau} e^{-\lambda t} \eee^x\left[h(X_{t})\right]dt,\;T<\infty.
\end{equation}
\end{small}
\begin{thm}\label{thm: conv_t}
    Let $T\in(0, \infty)$ and $\lambda>w$ be fixed. Then, $\mathcal{E}_2:=\|\lambda^2 \rr_{\lambda,T}^\tau  - \lambda I  -\L_\lambda^\tau\|\preceq  \lambda e^{-\lambda T}$ on $\mathcal{C}^2(\R^n)$. 
\end{thm}
\begin{proof}
    For any $T\in(0, \infty)$ and $h\in C^2(\R^n)$, we have the following dynamic programming 
    \begin{small}
    \begin{equation}\label{E: dp}
    \begin{split}
       &\rr(\lambda;\L^\tau) h(x)\\
         = &\int_0^T e^{-\lambda t} \eee^x\left[h(X_t)\cdot\exp\left\{-\int_0^t c(X_s)ds\right\}\right]dt\\
        & + \int_T^\infty e^{-\lambda t} \eee^x\left[h(X_t)\cdot\exp\left\{-\int_0^t c(X_s)ds\right\}\right]dt\\
        = & \int_0^T e^{-\lambda t} \eee^x\left[h(X_t)\cdot\exp\left\{-\int_0^t c(X_s)ds\right\}\right]dt\\
        & +e^{-\lambda T}\rr(\lambda;\L^\tau)\eee^x\left[h(X_T)\cdot\exp\left\{-\int_0^T c(X_s)ds\right\}\right],
    \end{split}
\end{equation}
    \end{small}

\noindent where the last line is by a standard change-of-variable argument. Note that \begin{small}
    \begin{equation}\label{E: trunc}
    \begin{split}
        \rr_{\lambda,T}^\tau h(x)&=  
         \int_0^{T\cj\tau} e^{-\lambda t} \eee^x\left[h(X_{t})\right]dt
        \\
        & = \int_0^T e^{-\lambda t} \eee^x\left[h(X_{t\cj\tau})\right]dt\\
        & = \int_0^T e^{-\lambda t} \eee^x\left[h(X_t)\cdot\exp\left\{-\int_0^t c(X_s)ds\right\}\right]dt.
    \end{split}
\end{equation}
\end{small}

\noindent On the other hand, we have $\|\rr(\lambda;\L^\tau)\|\preceq\frac{1}{\lambda}$ \cite{pazy2012semigroups}. Combining this fact, \eqref{E: dp}, and \eqref{E: trunc}, we have $\sup_{x\in\X}|\rr_{\lambda,T}^\tau h(x)-\rr(\lambda;\L^\tau) h(x)|\leq \sup_{x\in\X}|e^{-\lambda T} \rr(\lambda;\L^\tau)\eee^x[h(X(T))]|\preceq \frac{e^{-\lambda T}}{\lambda}$. Hence, $\|\rr_{\lambda,T}^\tau-\rr(\lambda;\L^\tau)\|\preceq \frac{e^{-\lambda T}}{\lambda}$, and the conclusion follows immediately. 
\end{proof}

Given that the error $\mathcal{E}_2\ll\mathcal{E}_1$, defining
\begin{equation}\label{E: finite_horizon}
    \L^\tau_{\lambda, T}:=\lambda^2 \rr_{\lambda,T}^\tau  - \lambda I
\end{equation}
yields $\L^\tau_{\lambda, T}h(x)\rightarrow\L_\lambda h(x)$ for all $h\in C^2(\R^n)$ and all $x\in\X$. 

\begin{rem}
    Note that using $\L^\tau_{\lambda, T}$ as an approximator for $\L_\lambda$ provides a convergence guarantee for all $h\in C^2(\R^n)$ restricted to $\X$, thereby relaxing the boundedness requirement from \eqref{E: resolvent}. The key step involves evaluating the operator up to the random first exit time $\tau$. 
\Qed
\end{rem}

By formula \eqref{E: finite_horizon}, the observation up to the stopping time $\tau$ is essential. While one might consider filtering on $\set{t\geq \tau}$ and working exclusively with non-exiting sample paths, this approach is generally infeasible, even when the measure of non-exiting paths is large. To show this, we separate $\rr^\tau_{\lambda, T}$ by $\int_0^T e^{-\lambda t} \eee^x\left[h(X_t)\mathds{1}_{\{t<\tau\}}\right]dt+ \int_0^T e^{-\lambda t} \eee^x\left[h(X_\tau)\mathds{1}_{\{t\geq\tau\}}\right]dt$, where the first term is equal to $\int_0^{T} e^{-\lambda t} \eee^x\left[h(X_t)| t<\tau]\cdot \ppp^x[t<\tau\right]dt$. In the most favorable case where $\ppp^x[t<\tau]$ is uniformly close to $1$ for all $t<T$, we tend to employ $\widetilde{\rr}_{\lambda, T}^\tau:=\int_0^{T} e^{-\lambda t} \eee^x[h(X_t)| t<\tau]dt$
rather than $\rr^\tau_{\lambda, T}$ as the approximator for $\rr(\lambda;\L^\tau)$. In this case, the error term $\widetilde{\mathcal{E}} =\int_0^T e^{-\lambda t} \eee^x\left[h(X_\tau)\mathds{1}_{\{t\geq\tau\}}\right]dt$, and
\begin{small}
    \begin{equation*}
\begin{split}
|\widetilde{\mathcal{E}}h(x)|
          &\leq  \sup_{x\in\X}|h(x)|\int_0^Te^{-\lambda t}\ppp^x[\tau\leq t]dt \leq \|h\| \frac{1-e^{-\lambda T}}{\lambda},\;\forall x\in\X.
\end{split}
\end{equation*}
\end{small}

 \noindent Now let $p_x(t)$ be the density of $\ppp^x[\tau\leq t]$ and $M_x(T):=\sup_{t\in[0, T]}|p_x(t)|$. We have  
 \begin{small}
   \begin{equation*}
\begin{split}
        |\widetilde{\mathcal{E}}h(x)|
        & \geq \inf_{x\in\partial\X}|h(x)|\left[\underbrace{\frac{1-e^{-\lambda T}\ppp^x[\tau\leq T]}{\lambda}}_{:=\widetilde{\epsilon}_1(x)}-\underbrace{\frac{1}{\lambda}\int_0^Te^{-\lambda t}\rho_x(t)dt}_{:=\widetilde{\epsilon}_2(x)}\right],
\end{split}
\end{equation*}   
 \end{small}

 \noindent where $\widetilde{\eps}_1(x)\geq \frac{1-e^{-\lambda T}}{\lambda}$ and
 $\widetilde{\eps}_2(x)\leq \frac{M_x(T)}{\lambda}\int_0^Te^{-\lambda t}dt=M_x(T)\cdot \frac{1-e^{-\lambda T}}{\lambda}$. This implies 
 \begin{small}
      $\|\widetilde{\mathcal{E}}\|=\max\{\frac{1-e^{-\lambda T}}{\lambda},  |1-M_x(T)|(\frac{1-e^{-\lambda T}}{\lambda})\}.$ 
 \end{small}
However, in terms of Yosida approximation,  replacing $\rr_{\lambda, T}^\tau$ with the  $ \widetilde{\rr}_{\lambda, T}^\tau$, the error is of the scale $\|\lambda^2\widetilde{\mathcal{E}}\|$, which is not convergent as $\lambda\rightarrow\infty$ even when $M_x(T)\approx 0$ uniformly.

\begin{rem}
    The aforementioned limitation particularly holds for unknown systems with ambiguous structure of  $p_x(t)$, where our current estimation cannot be improved  without additional system information. However, for systems admitting moment exponential stability, particularly those possessing a Lyapunov function $V(x)=|x|^p$ ($p\geq 2$) such that $\L V(x)\leq -qV(x)$ for some $q>0$, we have $\eee^x[V(X_t)]\preceq V(x)e^{-qt}$. Then, $\ppp^x[\tau\leq t]  = \ppp^x\left[\sup_{0\leq s\leq t}|X_s|\geq R\right]  \leq \frac{\eee^x\left[\sup_{0\leq s\leq t}|X_s|^p\right]}{R^p}   \preceq \frac{e^{-qt}}{R^2}$, whence $\lambda^2\widetilde{\mathcal{E}}h(x)\preceq\frac{\lambda^2}{R^p}\int_0^Te^{-(\lambda+q) t}dt\preceq\frac{\lambda^2}{R^p}\frac{1-e^{-(\lambda+q)t}}{\lambda + q} \downarrow 0$ for sufficiently large $R$ or $p$. \Qed
\end{rem}

\subsection{Finite-Rank Approximation of Yosida Approximation}

Based on \eqref{E: R_trunc} and \eqref{E: finite_horizon}, it suffices to prove that for any fixed $T\in(0, \infty)$, the operator $\rr_{\lambda, T}^\tau$ admits a finite-rank representation. This property naturally supports a machine learning approach using a finite dictionary of observable test functions as basis functions.

Note that verifying the finite-rank representability of this operator follows the same procedure as in \cite[Section IV.B]{meng2024resolvent}. We therefore omit the detailed argument. In summary, the semigroup $\set{\K_t^\tau}_{t\geq 0}$ admits a compact approximation, as do the operator $\rr^\tau_{\lambda, T}$ and $\L_{\lambda, T}^\tau$. Leveraging this compactness property, we obtain, for any fixed $T>0$, there exists a finite-dimensional approximation of the form $\L_{\lambda, T}^{\tau, N}=\lambda^2\rr_{\lambda,T}^{\tau,N}-\lambda I$ such that $\mathcal{E}_3h:=\|\L_{\lambda, T}^{\tau, N}h-\L_{\lambda, T}^{\tau}h\|\downarrow 0$ as $N\rightarrow\infty$.

\subsection{Martingale Problem and Weak Convergence}\label{sec: martingale}

In contrast to deterministic systems, where generator approximations yield uniform convergence of solutions, we demonstrate how the above generator approximation for stochastic systems relates to their weak solutions. Due to page limitations, we provide only a sketch of the  analysis. 

Let $\widehat{\L}^\tau:=\L_{\lambda, T}^{\tau, N}$, and let $\widehat{X}^\tau$ denote the weak solution of the system $(\widehat{f}, \widehat{b})$ generated by $\widehat{\L}$, i.e., 
$d\widehat{X}^\tau(t)=\widehat{f}(\widehat{X}^\tau(t))+\widehat{b}(\widehat{X}^\tau(t))dW(t)$. 
 Let $\widehat{\ppp}_{\lambda, N}:=\ppp_{\lambda, T}^{\tau, N}$ be the probability law of $\widehat{X}^\tau$. Then, for any fixed $T$, as $\lambda\rightarrow\infty$ and  $N\rightarrow\infty$, we claim that   $\widehat{\ppp}_{\lambda, N}$ converges
weakly to the measure to $\ppp^\tau$,   which is the law of the stopped solution $X^\tau$ to \eqref{eq1}. This claim is equivalent to $\widehat{X}^\tau$ converge in probability law to $X^\tau$. 

The proof falls in standard procedures of the martingale problem. It is clear that $\{\widehat{\ppp}_{\lambda, N}\}$  is a tight family of probability measures given the compact state space $\X$. 
Let $h\in C^2(\R^n)$ be any test function. Then we can show that 
the process $\{\widehat{M}_{\lambda, N}(t)\}_{t\geq 0}$ is a local martingale, where
\begin{equation}
    \begin{split}
        \widehat{M}_{\lambda, N}(t):=& h(\widehat{X}(t\cj\tau))-h(\widehat{X}(0))\\-& \int_0^{t\cj \tau}\L h(\widehat{X}(s))ds+\widehat{E}_h(t\cj \tau),
    \end{split}
\end{equation}
$\widehat{E}_h(t)=\int_0^t(Lh-\widehat{L}h)(\widehat{X}(s))ds$,  and, for all $h\in C^2(\R^n)$,  
$\widehat{\eee}_{\lambda, N}[\sup_{t\in[0, \tau\cj T]}\|\widehat{E}_h(t)\|]\preceq \sum_{i=1}^3\mathcal{E}_i\downarrow 0$ as $\lambda, N\rightarrow\infty$. 
Therefore, for any $0\leq r_1<r_2<...<r_n\leq s<t$ and $\{\varphi_j;\;j=1,2,...,n\}\subset C(\X)$, we alternatively have 
 \begin{small}
     \begin{equation}
    \widehat{\eee}_{\lambda, N}\left[\{\widehat{M}_{\lambda, N}(t)-\widehat{M}_{\lambda, N}(s)\}\prod\limits_{j=1}^n\varphi_j(X(r_j))\right]=0
\end{equation}
 \end{small}
 
\noindent We also define the   process 
\begin{equation}\label{E: martingale-limit}
M(t)=h(\widehat{X}(t\cj\tau))-h(\widehat{X}(0))-\int_0^{t\cj\tau}\mathcal{L} h(\widehat{X}(s))ds.
\end{equation}
 By the tightness of  $\{\widehat{\ppp}_{\lambda, N}\}$  on $\X$, we can find a convergent (weakly) subsequence $\widehat{\ppp}^{n} \rightarrow\ppp^\tau$  as $n\rightarrow\infty$ (where $(\lambda_n, N_n)\rightarrow \infty$ along the subsequence). We can also justify that $\{\widehat{M}_{\lambda, N}(t)\}_{t\in[0,T]}$ is uniformly  integrable. Therefore, by 
 the convergence of $\widehat{\eee}_{\lambda, N}[\sup_{t\in[0, \tau\cj T]}\|\widehat{E}_h(t)\|]$, we have
\begin{small}
     \begin{equation}
\begin{split}
     &\eee^\tau\left[\{M(t)-M(t)\}\prod\limits_{j=1}^n\varphi_j(X(r_j))\right]\\
      =& \lim\limits_{n\rightarrow\infty} \widehat{\eee}^{n}\left[\{M(t)-M(s)\}\prod\limits_{j=1}^n\varphi_j(X(r_j))\right] \\
      =&\lim\limits_{n\rightarrow\infty} \widehat{\eee}^{n}\left[\widehat{M}_{\lambda, N}(t)-\widehat{M}_{\lambda, N}(s)\}\prod\limits_{j=1}^n\varphi_j(X({r_j}))\right] =0.
\end{split}
\end{equation}
\end{small}

\noindent This means every limit of $\widehat{\ppp}^{n}$ solves the martingale problem w.r.t. \eqref{E: martingale-limit}. 
Note that under    local Lipschitz continuity conditions, the Yamada-Watanabe theorem guarantees uniqueness of solutions to the martingale problem. Consequently, every limit point is unique and must coincide with $\ppp^\tau$, which proves the claim from the preceding paragraph.

\section{Data-driven Algorithm}\label{sec: data}
Similar to the learning of Koopman operators \cite{williams2015data, mauroy2019koopman, meng2023learning}, obtaining a fully discretized version  $L$ of  the bounded linear operator $\widehat{\L}^\tau:=\L_{\lambda, T}^{\tau, N}=\lambda^2\rr_{\lambda,T}^{\tau,N}-\lambda I$ based on the training data typically relies on the selection of a discrete dictionary $ Z_N(x) = [z_1(x), z_2(x), \dots, z_N(x)]^\trans$ of   test functions, where  \( z_n \in C^2(\X)\) for each $n$. We adopt the   EDMD framework and refer to this learning procedure, which employs the approximation $\widehat{\L}^\tau$, as resolvent-type EDMD (RT-EDMD).

Let $(\alpha_i, \phi_i)_{i=1}^{N}$ be the eigenvalues and eigenfunctions of $\L$. Let $(\beta_i, \varphi_i)_{i=1}^{N}$ be the eigenvalues and eigenvectors of $L$. Then, the following are expected to hold: 1) For any $h\in\operatorname{span}\{z_1, z_2, \cdots,z_{N}\}$ such that $h(x)=Z_N(x)\operatorname{\theta}$ for some column vector $\theta$, we have that $\L h(x)\approx \widehat{\L}^\tau h(x) \approx Z_N(x)(L\theta)$ for all $x\in\X$. 
2) $\alpha_i\approx\beta_i$ and $\phi_i(x)\approx Z_N(x)\varphi_i$ for each i. 

The learning procedure for such an \( L \) follows \cite{meng2024resolvent}, with two key data-related modifications for stochastic systems: 1) for each sample path \( \omega \), we approximate \( X(\tau, \omega) \) at \( \partial\X \); and 2) we use empirical averages over sample paths instead of the true expectation.

The procedure is detailed as follows. We first select $m$ initial states $\{x_i\}_{i=1}^m$. For each $X(0)=x_i$, we generate $J$ independent sample path $\set{\omega_j}_{j=1}^J$ from \eqref{eq1}, and denote the corresponding solutions by $X_{i,j}(t):=X(t, \omega_j)$ satisfying $X_{i,j}(0)=x_i$ for all $j$. We collect discrete-time observations (snapshots) at a rate of $\gamma$ Hz over the interval $[0, T]$  for all sample path, yielding the total number of snapshots as $\Gamma=\gamma T +1$. The observation instants are defined as $t_k=k T/\Gamma$ for $k\in\set{0, 1, \cdots, \Gamma}$.  Then,  each sample path is obtained through the Euler–Maruyama integration scheme at $\{t_k\}_{k=0}^\Gamma$, and represented as $X_{i,j} = \left[ x_i, X_{i,j}(t_1),\, X_{i,j}(t_2),\, \dots,\, X_{i,j}(t_\Gamma) \right]^\trans$ for each $j$. Similarly, we denote the stopped sample path as $X_{i,j}^\tau = \left[ x_i, X_{i,j}(t_1\cj\tau),\, X_{i,j}(t_2\cj\tau),\, \dots,\, X_{i,j}(t_\Gamma\cj\tau) \right]^\trans$. We can omit the index $j$ to emphasize the entire process. 

\begin{rem}
For each \( j \), we track whether \( X_{i,j}(t_k) \in \X \) to obtain the stopped path. The discrete exit time is recorded as \( k^* = \inf\{k : X_{i,j}(t_k) \notin \X\} \). If \( k^* < \Gamma \), we linearly interpolate between \( X_{i,j}(t_{k^*-1}) \) and \( X_{i,j}(t_{k^*}) \) to approximate the intersection with \( \partial\X \), yielding \( X_{i,j}(t_k \wedge \tau) \) for \( k \geq k^* \). See \cite[Algorithm 1]{meng2023learning} for details. \Qed
\end{rem}

Next, we summarize how the data is stacked for learning $L$. As a standard treatment, we stack the feature data   as  $ \mathbf{X} = [ Z_N(x_1), Z_N(x_2), \cdots, Z_N(x_m)]^T$. To construct the label stack at each time $t_k$, we proceed as follows.

The observation dictionary $\overline{Z}_N$ applied to the sample path   at time $t_k$ starting from initial point $x_i$ is represented as:
\begin{small}
    \begin{equation}
\overline{Z}_N(X_i^\tau(t_k)) = \frac{1}{J} \sum_{j=1}^{J} Z_n(X_{i,j}(t_k\cj\tau)) .
\end{equation}
\end{small}

\noindent The core step in constructing the label data stack involves numerically evaluating the resolvent integral. To do this, we construct an intermediate data matrix representing the integrands at discrete time $\set{t_k}$:
\begin{equation}
U_i = \lambda^2 e^{-\lambda t_k} \overline{Z}_N(X_i^\tau(t_k)).
\end{equation}

Based on the Yosida-like approximation $\lambda^2\rr_{\lambda,T}^{\tau,N}-\lambda I$, we construct the   matrix $\mathbf{Y}_\lambda (=\mathbf{Y}^{\Gamma}_{\lambda}) = \mathbf{I}^{\Gamma}_{\lambda} - \lambda \mathbf{X}$, where $\mathbf{I}^{\Gamma}_{\lambda}$ is the integration matrix approximating $\lambda^2\rr_{\lambda,T}^{\tau,N}$ constructed by
\begin{small}
    \begin{equation}
\mathbf{I}^{\Gamma}_{\lambda} = \begin{bmatrix} 
\mathcal{T}(U_1[:,1]) & \cdots & \mathcal{T}(U_1[:,N]) \\
\vdots & \ddots & \vdots \\
\mathcal{T}(U_m[:,1]) & \cdots & \mathcal{T}(U_m[:,N])
\end{bmatrix}
\end{equation}
\end{small}

\noindent with some linear operator \(\mathcal{T}(\cdot)\)  that approximates an integral using the trapezoidal rule. 

To retrieve $L$ from $\mathbf{X}$ and $\mathbf{Y}_\lambda$, the following procedures are identical to that in \cite{meng2024resolvent}, which are summarized below. 

\textit{A. Matrix generator estimation.}
Here we estimate the generator matrix through a least-squares method: $L = \arg\min_{A \in \mathbb{R}^{N \times N}} \|\mathbf{Y}^{\Gamma}_{\lambda} - \mathbf{X}A\|_F$
with the closed-form solution $L = (\mathbf{X}^\trans\mathbf{X})^{\dagger}\mathbf{X}^\trans\mathbf{Y}^{\Gamma}_{\lambda}$ \cite{williams2015data}.  Through this step, we obtain a finite-dimensional representation of the   generator of using the selected dictionary, which satisfies approximation properties  1) and  2) stated at the beginning of this section.

\textit{B. Modification under low sampling rate constraints.}
For low sampling rate scenarios, we introduce a modification based on the resolvent identity. We define the modified generator matrix: $L_{mod} = A^{\dagger}B$, 
where
\begin{small}
    \begin{equation}
A = \frac{\lambda-\mu}{\mu^2}\mathbf{I}^{\Gamma}_{\mu} + \mathbf{X}, \quad B = \frac{\lambda}{\mu}\mathbf{I}^{\Gamma}_{\mu} - \lambda \mathbf{X},
\end{equation}
\end{small}

\noindent $\mathbf{I}^{\Gamma}_{\mu}$ is computed with small $\Gamma$ using a small parameter $\mu$. This modification is necessary when $\Gamma$ is small and $\lambda$ must be large for the theoretical convergence of $\widehat{\L}$. However, to guarantee numerical precision in the integration, $\Gamma$ and $\lambda$ cannot be both large. Thus, without the modification, it is not feasible to resolve this conflict. This modification is achieved by applying the \textit{first resolvent identity}. For  details, we refer the reader to \cite[Appendix B]{meng2024resolvent}.

We provide convergence analysis for the data-driven approximation of $\L$ in the Appendix \ref{sec: app}, covering both scenarios using empirical averages and cases with small noise where a single sample path is used to replace the mean value.



\section{Numerical Examples} \label{sec:4}

In this section, we consider two classical theoretical systems with known spectra to validate the effectiveness and improved performance of the proposed method through comparisons with standard approaches including EDMD, and generator EDMD (gEDMD). 

Through learning the operator $\L$, we can achieve complete system identification. For   projection functions $\Psi_i(x)=x_{(i)}$ on the the $i$-th dimension of $x$, the second derivatives vanish, yielding $\mathcal{L}\Psi_i(x) = f_{(i)}(x)$
where $f_{(i)}(x)$ is the $i$-th component of the drift term. For identifying the diffusion coefficients, we utilize quadratic basis functions $\Phi_{i,j}(x)=x_{(i)}x_{(j)}$. The generator action gives $\mathcal{L}\Phi_{i,j}(x) = f_{(i)}(x)x_{(j)} + f_{(j)}(x)x_{(i)} + (b(x)b(x)^\trans)_{ij}$. We can then    isolate the diffusion term through  $b(x)b(x)^\trans)_{ij} = \mathcal{L}\Phi_{i,j}(x) - \left[f_{(i)}(x)x_{(j)} + f_{(j)}(x)x_{(i)}\right]$. 

We adopt the mean error between sample paths generated by the learned operator and true paths (under identical noise realizations) as our performance metric. If the analytical eigenvalues are known,  we can adopt the mean absolute error (MAE) as the evaluation metric for the accuracy of spectrum estimation. We define $\text{MAE} = \frac{1}{n} \sum_{i=1}^{n} |\alpha_i - \hat{\alpha}_i|$, 
where $\alpha_i$ denotes the $i$-th true eigenvalue, $\hat{\alpha}_i$ is the corresponding estimated eigenvalue, and $n$ is the number of eigenvalues.

The dictionary selection is highly flexible and can include polynomial bases, radial basis functions, 
or neural network features, enabling the algorithm to adapt to various nonlinear systems without prior knowledge. In non-degenerate elliptic diffusion settings, we select either polynomials of sufficient order 
that enable approximation closure of second-order differential operators at semi-global scales.

\subsection{Ornstein-Uhlenbeck Process}

We consider the classical one-dimensional Ornstein--Uhlenbeck (OU) process defined by the following SDE:
\begin{equation}
\mathrm{d}X(t) = \mu X(t)\,\mathrm{d}t + \sigma\,\mathrm{d}W(t),\;\;\mu < 0,\;\sigma > 0.
\label{eq:ou}
\end{equation}
We select $\mu = -0.5,\ \sigma =0.02$ for analytical investigation. Under the Koopman operator framework, this system exhibits an explicit spectral structure\cite{pavliotis2014stochastic}\cite{gaspard1995spectral}, with eigenvalues given by $\lambda_n = n\mu$, for $n\in\mathbb{N}_0$.
These eigenvalues form a degenerate spectrum, which consists of integer multiples of the drift coefficient \(\mu\), reflecting the intrinsic linear Gaussian structure of the OU process.

To validate the effectiveness of the proposed RT-EDMD algorithm, we first conduct a preliminary experiment under a sampling frequency of \SI{100}{Hz}. Specifically, we uniformly select $m = 200$ initial points from the interval $[-1, 1]$, and for each point, randomly generate $J = 100$ sample trajectories. The observable functions are chosen as monomials $z_n(x) = x^n$ for $n = 1, \ldots, N$, with $N = 5$. 
The RT-EDMD algorithm accurately estimates the eigenvalues of the Koopman operator under this setting. The estimation error, measured by the MAE, is $1.5 \times 10^{-3}$.

To evaluate robustness across sampling frequencies, we perform $100$ independent trials per sampling frequency setting. Fig.~\ref{fig：1d_freq} compares both spectral estimation methods across sampling frequencies, with the horizontal axis showing frequency (Hz) and the vertical axis displaying the MAE of Koopman generator spectrum estimates.
The boxplots in the figure illustrate the distribution of estimation errors over multiple independent trials for each method and frequency. The boxes indicate the interquartile range (IQR), with the horizontal line inside each box representing the median error. The whiskers extend to the most extreme data points within 1.5 times the IQR from the quartiles, and the dots denote outliers. In Fig.\ref{Fig3.sub1}, RT-EDMD consistently exhibits the lowest median error across all sampling frequencies. Even at small sampling intervals, it remains stable in capturing the system’s average dynamical behavior, effectively mitigating error amplification induced by diffusion perturbations. As the sampling frequency increases beyond 50~Hz, the estimation error plateaus with controlled variability, demonstrating strong generalization ability and robustness.

In contrast, EDMD shows a clear upward trend in error as the sampling frequency increases, which is closely related to its construction via the Koopman logarithm method (KLM)\cite{drmavc2021identification}\cite{klus2020data}. Specifically, EDMD approximates the finite-time Koopman operator using least squares and then applies a logarithmic transformation to obtain the continuous-time generator\cite{mauroy2019koopman} via $
\L = \frac{1}{t} \log(\K_t).
$
In SDE systems, higher sampling frequencies render the logarithmic transformation numerically sensitive. Meanwhile, the stochastic diffusion term becomes more significantly relative to the drift component—scaling as $\mathcal{O}(\sqrt{\Delta t})$ versus $\mathcal{O}(\Delta t)$, respectively. This leads to a declining signal-to-noise ratio as frequency increases, resulting in ill-conditioning of the Koopman operator estimate and ultimately causing a systematic increase in spectral estimation error.

As for gEDMD, which follows a finite difference method (FDM)\cite{bramburger2024auxiliary} \cite{nejati2021data} paradigm, approximates the generator directly using
$
\frac{\K_t - I}{t} \Rightarrow \L.
$
By incorporating a weighted inner product and leveraging derivative information of the observables, gEDMD partially alleviates high-frequency sampling issues. In Fig.\ref{Fig3.sub3}, it achieves relatively stable and low errors in the mid-frequency range ($20–100$Hz), although its performance still falls short of the proposed RT-EDMD. At higher frequencies (e.g., $200$Hz), error amplification re-emerges, indicating that the ability of gEDMD to suppress diffusion-induced disturbances remains limited.


\begin{figure*}[thpb]
    \centering
    \subfloat[]{
        \label{Fig3.sub1}
        \includegraphics[width=.25\textwidth]{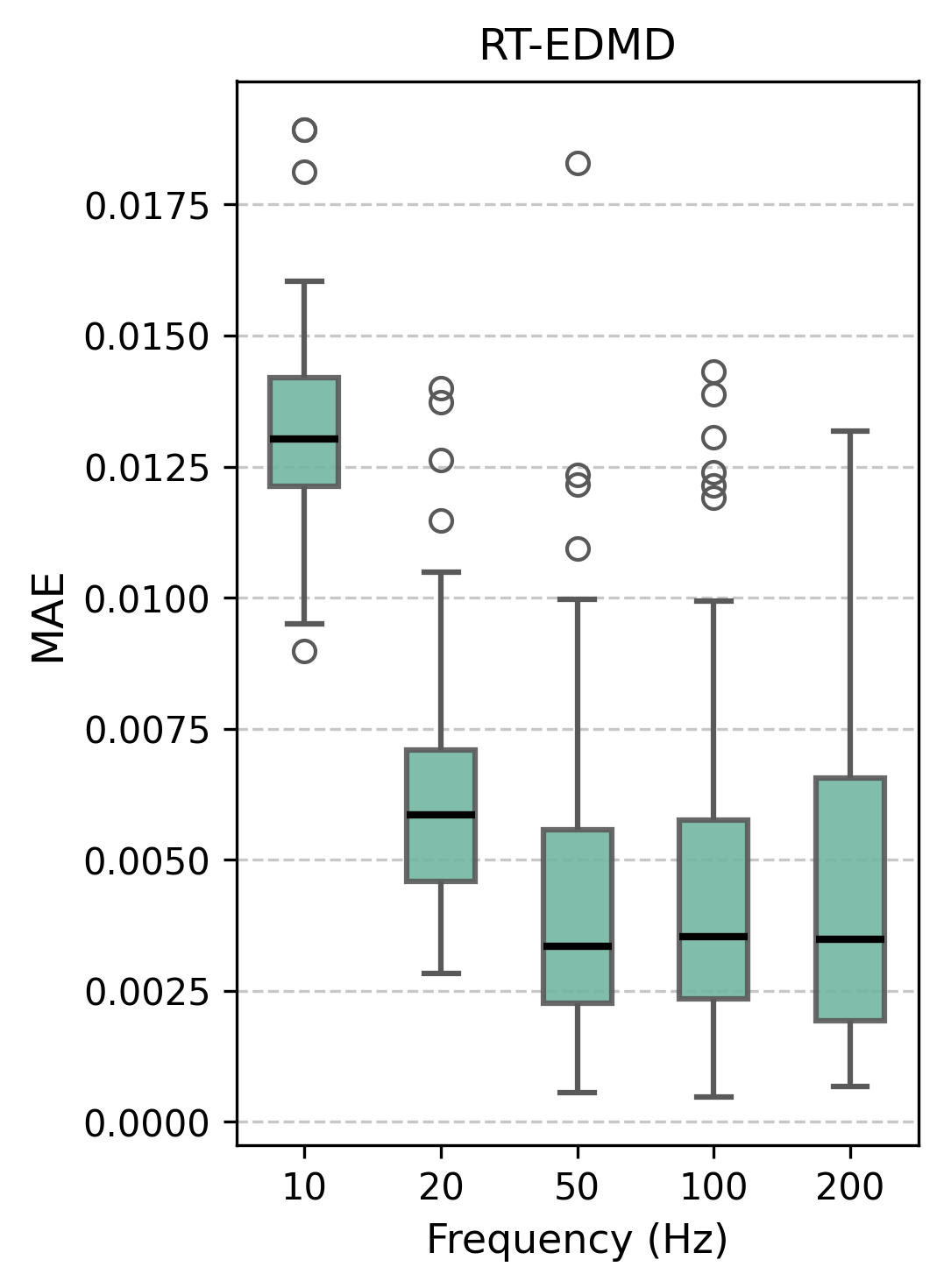}}
    \subfloat[]{
        \label{Fig3.sub2}
        \includegraphics[width=.25\textwidth]{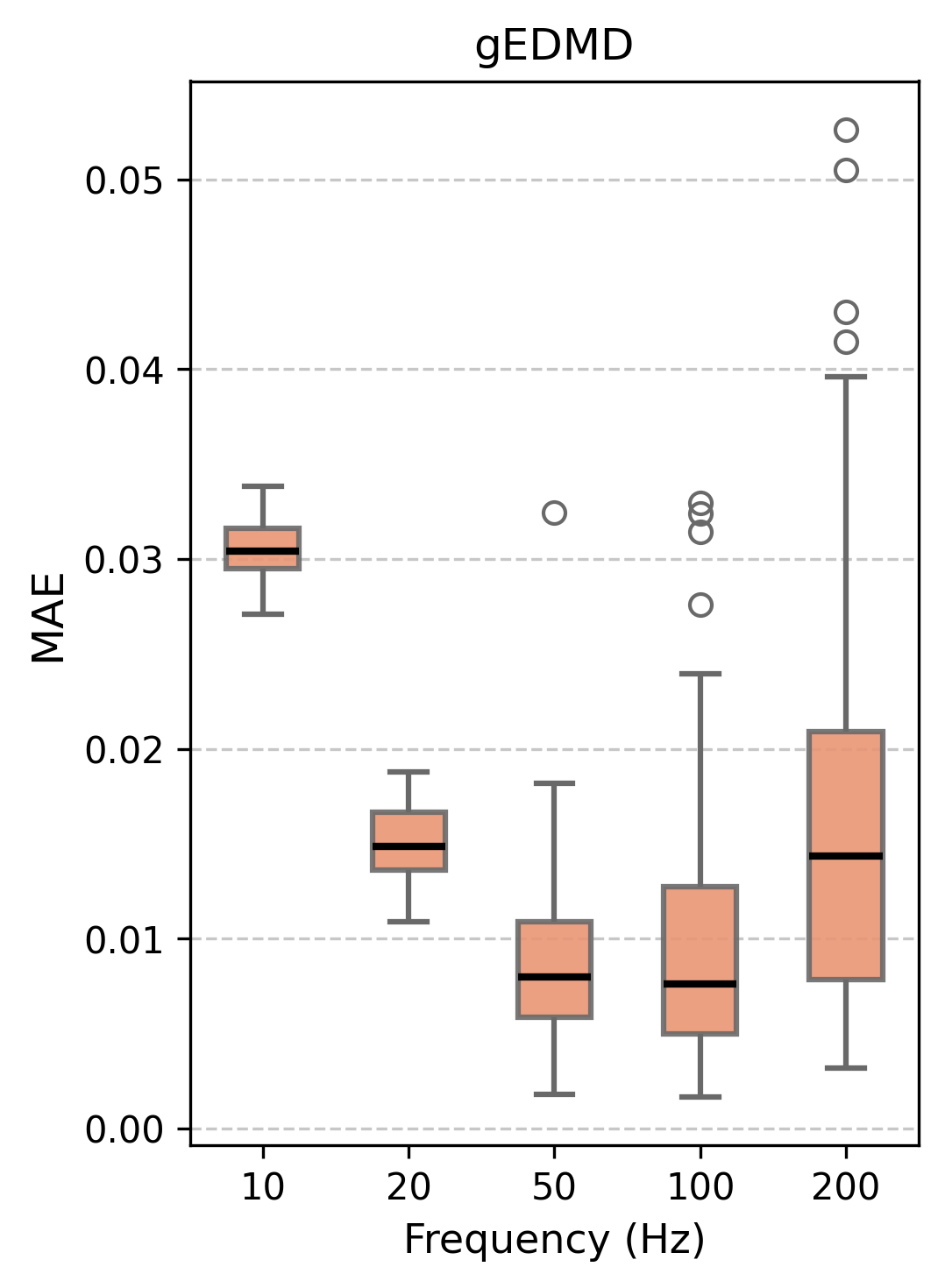}}
    \subfloat[]{
        \label{Fig3.sub3}
        \includegraphics[width=.25\textwidth]{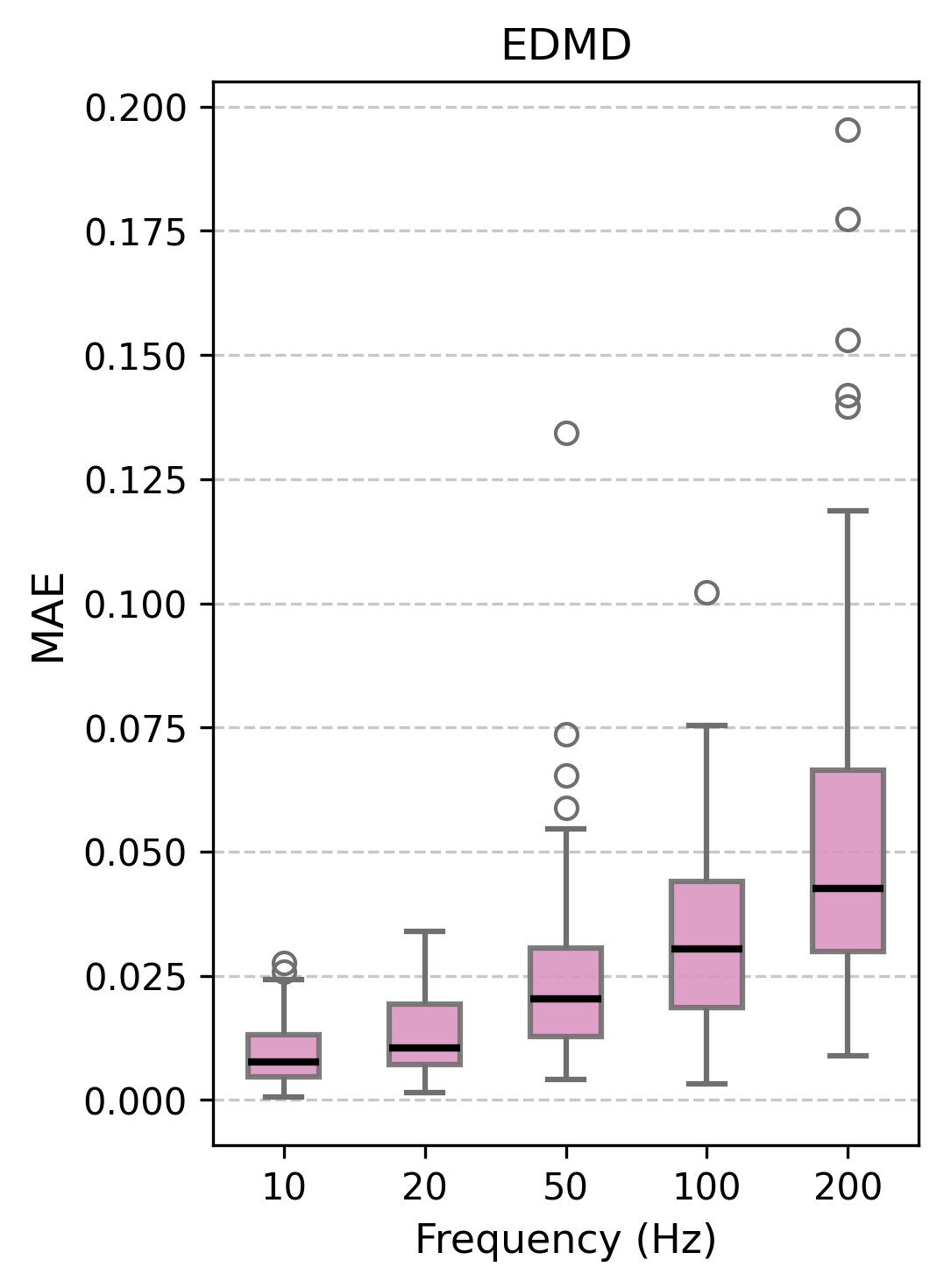}}	
    \caption{Spectral estimation error of the Ornstein–Uhlenbeck (OU) process at different frequencies}
    \label{fig：1d_freq}
\end{figure*}

Next, we validate the proposed RT-EDMD method’s capability in accurately capturing system dynamics under uncertainty. Under a sampling frequency of $100$ Hz, with $J = 100$ trajectories and $m = 200$ distinct initial conditions, we estimate the drift and diffusion coefficients of the system. As shown in Fig.~\ref{fig:re_path_1d}, the resulting parameter estimates are employed to reconstruct state trajectories from the OU process across various initial values, including both near-equilibrium and far-from-equilibrium cases. The upper subplots of Fig.~\ref{fig:re_path_1d} show that RT-EDMD (red dashed) closely approximates the true trajectories (blue solid), while EDMD and gEDMD show deviations and spurious oscillations, especially near the equilibrium. The lower subplots display the absolute errors on a logarithmic scale, where RT-EDMD consistently achieves errors between $10^{-3}$ to $10^{-5}$, demonstrating a substantial improvement over EDMD and gEDMD, whose error magnitudes typically lie between $10^{-2}$ and $10^{-1}$.  This comparative analysis highlights the superior precision in estimating system parameters and its robustness in capturing both short-term transitions and long-term mean-reverting behaviors of the OU process. These results underscore RT-EDMD's effectiveness in faithfully recovering SDE dynamics and delivering stable, low-error trajectory reconstructions across diverse initializations.

\begin{table}[htbp]
\centering
\caption{MAE (log scale) for Koopman spectrum error: single vs. multi-trajectory (100~Hz sampling frequency)}
\begin{tabular}{lcccc}
\toprule
\textbf{} & \textbf{RT-EDMD} & \textbf{gEDMD} & \textbf{EDMD} &  \\
\midrule
$J=1$      & $\mathbf{2.6 \times 10^{-3}}$ & $7.8 \times 10^{-3}$ & $2.5 \times 10^{-2}$ \\
$J=100$& $\mathbf{1.5 \times 10^{-3}}$ & $6.8 \times 10^{-3}$ & $2.3 \times 10^{-2}$  \\
\bottomrule
\end{tabular}
\label{table1}
\end{table}

In many practical applications, such as physical experiments, financial modelling, or online control, it is often difficult to collect many repeated samples from the same initial condition. To assess the proposed method under weak observability, we consider the limiting case of a single trajectory per initial point. The sampling frequency is fixed at 100~Hz, with all other settings unchanged. We compare the spectral estimation accuracy of various methods under both single- and multi-trajectory conditions. As shown in Table~\ref{table1}, the four methods yield notably different MAEs. RT-EDMD achieves the best performance in all cases, with an error of \(2.6 \times 10^{-3}\) in the single-trajectory case and \(1.5 \times 10^{-3}\) for multi-trajectory (\(J = 100\)).



\begin{figure*}
    \centering
    \includegraphics[width=1.05\textwidth]{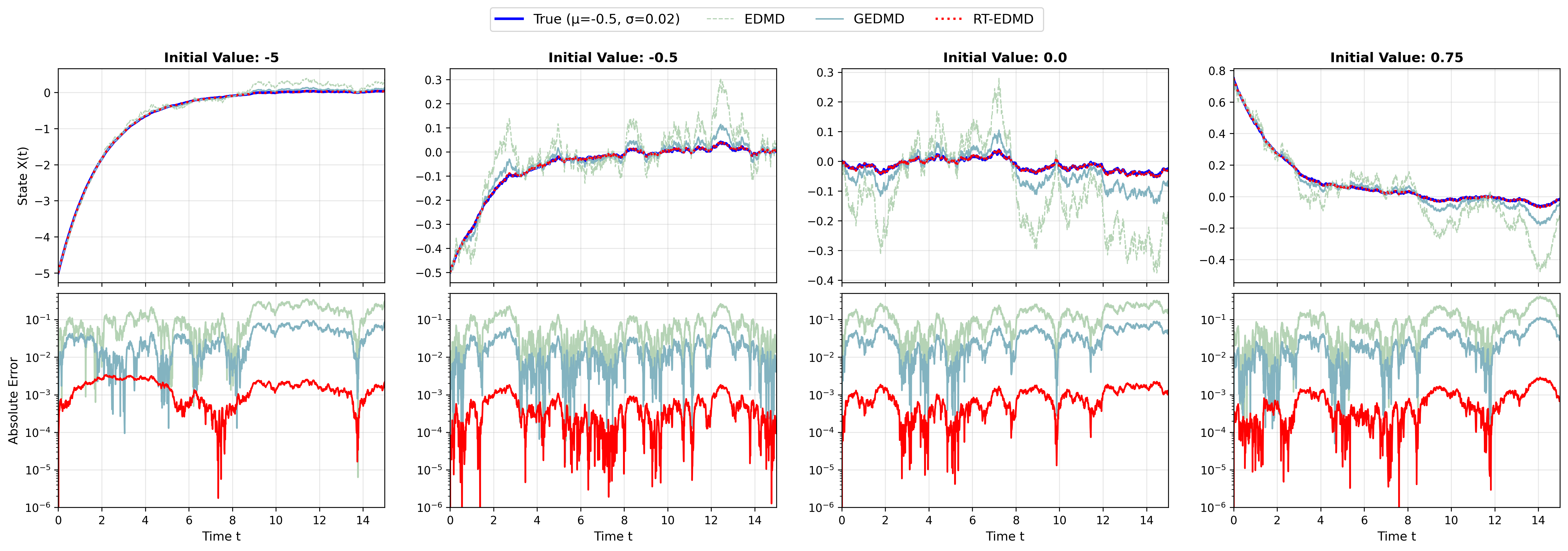}
    \caption{Reconstruction of sample paths from the OU process}
    \label{fig:re_path_1d}
\end{figure*}


\subsection {Noisy Lotka-Volterra (predator-prey) system}

Here we consider the classical Lotka-Volterra competition system, which is extensively applied in ecology and biology to model predator-prey interactions and competitive dynamics between two species. To account for realistic stochastic perturbations frequently encountered in ecological systems, multiplicative noise terms are introduced. Thus, the resulting SDE system under investigation can be expressed as follows:
{\small
    \begin{equation}
\begin{aligned}
dX_1(t) &= (a_1 - b_1 X_2(t) - c_1 X_1(t)) X_1(t)\,dt + \sigma_1 X_1(t)\,dW_t^{(1)}, \\
dX_2(t) &= (-a_2 + b_2 X_1(t) - c_2 X_2(t)) X_2(t)\,dt + \sigma_2 X_2(t)\,dW_t^{(2)}.
\end{aligned}
\end{equation}
}

Here, \( X_1, X_2 \) denote the population sizes of the two species. The model parameters \(a_1, b_1, c_1, a_2, b_2, c_2 > 0\) are determined by the ecological characteristics of the specific species considered. Additionally, \(\sigma_1, \sigma_2 \geq 0\) are nonnegative parameters that characterize the intensity of environmental fluctuations or other uncertainties affecting the dynamics of the populations.

The system parameters in this study are configured as follows: 
$a_1 = 1.0$, $b_1 = 0.5$, $c_1 = 0.01$; $a_2 = 0.75$, $b_2 = 0.25$, $c_2 = 0.01$, 
with noise intensity parameters $\sigma_1 = \sigma_2 = 0.05$. 
According to \cite{arato2003famous}, the principal eigenvalues of the system 
are given by $\lambda^{1,2} = -0.02509 \pm 0.86363i$. 



To validate the effectiveness and robustness of our method, we compare the principal spectrum estimation performance of different algorithms under varying sampling frequencies. In this case,we select 200 initial points randomly, each generating 100 independent trajectories for sampling. Table \ref{table2} presents the MAE comparison among EDMD, gEDMD, and RT-EDMDat different sampling frequencies. The results indicate that RT-EDMD consistently achieves the lowest MAE across all sampling rates. Although EDMD exhibits relatively stable performance when the sampling frequency changes, its overall errors remain higher than those of RT-EDMD and gEDMD. In summary, RT-EDMD demonstrates superior robustness and accuracy under both low and high sampling frequencies, effectively coping with noisy conditions in nonlinear systems. Hence, it shows significant advantages over the other two conventional methods in this experiment.
\begin{table}[htbp]
  \centering
  \caption{Mean Absolute Error (MAE) comparison of Noisy Lotka-Volterra system across sampling frequencies}
  \label{tab:mae_comparison}
  \begin{tabular}{lrrrr}
    \toprule
    Method & \multicolumn{4}{c}{Sampling Frequency (Hz)} \\
    \cmidrule(lr){2-5}
           & 10   & 20   & 50   & 100  \\
    \midrule
    RT-EDMD &  $\mathbf{0.0143}$ &  $\mathbf{0.0099}$ & $\mathbf{0.0131}$ & $\mathbf{0.0084}$ \\
    EDMD    & 0.0176 & 0.0173 & 0.0173 & 0.0183 \\
    gEDMD   & 0.0271 & 0.0126 & 0.0150 & 0.0121 \\
    \bottomrule
  \end{tabular}
  \vspace{0.5em}
  \footnotesize
  \label{table2}
\end{table}


\section{CONCLUSION}

In this paper, we have proposed a Koopman spectral analysis and system identification method for stochastic dynamical systems based on the Resolvent-Yosida approximation. By constructing finite-horizon and finite-rank Yosida approximations, the developed RT-EDMD algorithm robustly estimates the Koopman generator from observational data, enabling accurate extraction of spectral modes and reliable identification of system coefficients.

From a theoretical standpoint, rigorous mathematical guarantees regarding numerical stability and convergence of the proposed algorithm are established through analyses rooted in the martingale problem and weak convergence theory. Numerically, the effectiveness and robustness of RT-EDMD have been thoroughly demonstrated using representative examples, including the Ornstein–Uhlenbeck process and a noisy Lotka–Volterra predator-prey system. Experimental results indicate that, even under conditions of low sampling rates or single-trajectory observations, the RT-EDMD algorithm effectively captures essential spectral features, accurately identifies drift and diffusion coefficients, and consistently outperforms existing benchmark methods in terms of stability and performance.

\bibliographystyle{IEEEtran}
\bibliography{IEEEabrv,root}

\newpage
\appendices
\section{Convergence Analysis in Section \ref{sec: data}} \label{sec: app}
\subsection{Convergence when using empirical averages to replace the true average}\label{sec: lln}
Note that when replacing the $\eee^{x_i}[h(X(t_k\cj\tau)]$ in the resolvent integral with its empirical average $\sum_{j=1}^{J} h(X_{i,j}(t_k\cj\tau))/J$ at each $t_k$ for each $i$, we must ensure convergence via the law of large numbers (LLN). The precision is measured under the corresponding probability \footnote{In \cite{nejati2021data}, the authors used $\pp$, but  in our context it is recast to be $\mathsf{P}^x$. The uniqueness of $\mathsf{P}^x$ is by Kolmogrov's extension theorem. 
}  $\mathsf{P}^{x_i}:=\otimes_{j=1}^\infty\ppp^{x_i}$. However, to fit the Yosida-type approximation and the martingale problem argument, we need the precision to be measured in $L_1$ sense in stead of in probability as in \cite{nejati2021data}. 

We need to leverage the convergence in the $L_1$ sense, i.e.,
\begin{equation} \label{E: l1}
    \mathsf{E}^{x_i}\left|\frac{1}{J}\sum_{j=1}^J h(X_{i,j}(t_k\cj\tau))-\eee^{x_i}[h(X(t_k\cj\tau))]\right|\rightarrow 0
\end{equation}
for each $i$ and $k$. 
This is indeed the case as an existing result, even though it is seldom mentioned. 

We prove the $L_1$ convergence of LLN as in \eqref{E: l1} using the backward martingale argument. 
\begin{deff}[Backward martingale]
A backward martingale is a stochastic process $\{X_{-n}\}_{n=1,2,\cdots}$ such that, for each $n$, $X_{-n}$ is $L_1$ integrable and $\mathcal{F}_{-n}$-measurable, and satisfies
\begin{equation}
    \mathsf{E}[X_{-n-1}\;|\;\mathcal{F}_{-n}]=X_{-n}.
\end{equation}
\end{deff}

\begin{thm}\textbf{(Backward martingale convergence theorem):}
For every backward maringale, as $n\rightarrow \infty$,
\begin{equation}
    X_{-n}\rightarrow \mathsf{E}[X_{-1}\;|\;\mathcal{F}_{-\infty}]\;\;\mathsf{P}\text{-a.s. and in}\; L_1. 
\end{equation}
\end{thm}
\begin{thm}[Kolmogorov's $0$-$1$-law]
Let $\mathcal{F}_1,\mathcal{F}_2,\cdots$ be independent $\sigma$-fields and denote by $\mathcal{F}_{\infty}=\cap_{n=1}^\infty\sigma\left(\cup_{k=n}^\infty \mathcal{F}_k\right)$ the corresponding tail field. Then 
$$\mathsf{P}[A]\in\{0,1\},\;\;\forall A\in\mathcal{F}_\infty.$$
\end{thm}
\noindent\textbf{Proof of \eqref{E: l1}}: Let $Y_j=h(X_{i,j}^\tau(t))$ for $j\in\{1,2,\cdots\}$.
Then $\{Y_j\}$ is $L_1$ integrable and i.i.d. w.r.t. $\mathsf{P}^{x_i}$. Let $S_J=\sum_{j=1}^JY_j$ be the finite sum and let $\Sigma_{-J}=\frac{S_J}{J}$ be the average.  Then 
the $\sigma$-field $\mathcal{F}_{-J}=\sigma\{S_J,S_{J+1},\cdots\}$ is a decreasing filtration. Due to the independence of $\{Y_j\}$, we have
\begin{equation}\label{E: backward}
\begin{split}
       \mathsf{E}^{x_i}[\Sigma_{-1}\;|\;\mathcal{F}_{-J}]& =\mathsf{E}^{x_i}[Y_1\;|\;S_J,S_{J+1},\cdots]\\
       &=\mathsf{E}^{x_i}[Y_1\;|\;S_J,Y_{J+1},Y_{J+2},\cdots]\\
       & =\mathsf{E}^{x_i}[Y_1\;|\;S_J]. \end{split}
\end{equation}
 Notice that  $\mathsf{E}^{x_i}[Y_j\;|\;S_J]=\mathsf{E}^{x_i}[Y_\iota\;|\;S_J]$  by symmetry for $\iota,j\in\{1,\cdots,J\}$,
  then 
 \begin{equation}
     J\mathsf{E}^{x_i}[Y_\iota\;|\;S_J]=\sum_{\iota=1}^J\mathsf{E}^{x_i}[Y_\iota\;|\;S_J]=\mathsf{E}^{x_i}[S_J\;|\;S_J]=S_J.
 \end{equation}
Combining the above, we have
$$\mathsf{E}^{x_i}[\Sigma_{-1}\;|\;\mathcal{F}_{-J}]=\frac{S_J}{J}=\Sigma_{-J}, $$
which verifies that $\{\Sigma_{-J}\}$ is a backward martingale. By the backward martingale convergence theorem, we immediately have
$$\frac{S_J}{J}\rightarrow \mathsf{E}^{x_i}[Y_1\;|\;\mathcal{F}_{-\infty}],\;\;\mathsf{P}^{x_i}\text{-a.s. and in}\; L_1. $$
By Kolmogorov's $0$-$1$ law, we have that all $A$ in the tail field $ \mathcal{F}_{-\infty}$ have probability either $0$ or $1$, which in turn implies that the conditional expectation $\mathsf{E}^{x_i}[Y_1\;|\;\mathcal{F}_{-\infty}]$ must be a constant (by the definition of conditional expectation) and should be equal to the average $\mathsf{E}^x[Y_1]=\eee^{x_i}[Y_1]$. \Qed

Combining \eqref{E: l1} with the methodology from Section \ref{sec: data}, under high precision of numerical integration and data fitting, we can immediately obtain the approximation $\L h(x)\approx \widehat{\L}^\tau h(x) \approx Z_N(x)(L\theta)$ as stated.  The remainder of the convergence analysis follows the same procedure as in Section~\ref{sec: martingale}, but now under the measure $\mathsf{P}$ (which has broader measurability than $\ppp$). We therefore omit the repetition.

\subsection{Convergence   for single-sample-path estimation in the small-noise scenario}
Now we suppose that the diffusion term is small. In this case, we examine how sample paths converge to the mean trajectory. Let $\overline{(\cdot)}:=\eee(\cdot)$ for simplicity. Then, for \eqref{eq1}, we have $X(t)-\overline{X(t)}=\int_0^t [f(X(s))-\overline{f(X(s))}]ds+\int_0^t b(X(s))dW(s)$. Consequently, for any fixed $t$, 
\begin{equation}
    \begin{split}
        &\eee|X^\tau(t)-\overline{X^\tau(t)}|^2\\\preceq & \eee \int_0^t\left| f(X^\tau(s))-\overline{f(X^\tau(s))}\right|^2ds+\eee\left|\int_0^tb(X^\tau(t))dW(t)\right|^2\\
        \preceq & \|bb^\trans(X^\tau(t))\|\downarrow 0
    \end{split}
\end{equation}
where the first term is bounded by $\preceq \|bb^\trans(X^\tau(t))\|$ by LLN, and the second term is bounded by $\preceq \|bb^\trans(X^\tau(t))\|$ by It\^{o}'s  Isometry. By the same reasoning as in Section \ref{sec: lln}, we can immediately obtain the approximation $\L h(x)\approx \widehat{\L}^\tau h(x) \approx Z_N(x)(L\theta)$ as stated.  The remainder of the convergence analysis follows the same martingale problem procedure as in Section~\ref{sec: martingale}.
\end{document}